\documentclass[12pt]{article}

\usepackage{amssymb}
\usepackage{amsfonts}
\usepackage{amsmath}
\usepackage{amsthm}
\usepackage{graphicx}
\usepackage[boxruled]{algorithm2e}

\def\a{{\alpha}}

\def\g{{\gamma}}

\def\q{{\theta}}
\def\f{{\phi}}

\newtheorem{thm}{Theorem}
\newtheorem{lm}{Lemma}

\usepackage{xcolor}

\newcommand{\squeezelist}{\setlength{\itemsep}{0pt}}
\newtheorem{innercustomthm}{Theorem}
\newenvironment{customthm}[1]
  {\renewcommand\theinnercustomthm{#1}\innercustomthm}
  {\endinnercustomthm}
\usepackage{enumitem}
\usepackage{url}
\newcommand{\figlab}[1]{\label{fig:#1}}
\newcommand{\seclab}[1]{\label{sec:#1}}

\newcommand{\lemlab}[1]{\label{lem:#1}}
\newcommand{\thmlab}[1]{\label{thm:#1}}

\newcommand{\figref}[1]{\ref{fig:#1}}
\newcommand{\secref}[1]{\ref{sec:#1}}

\newcommand{\lemref}[1]{\ref{lem:#1}}
\newcommand{\thmref}[1]{\ref{thm:#1}}

\usepackage{centernot}

\begin{document}

\title{\textbf{Cut Locus Realizations\\
on Convex Polyhedra}}

\author{Joseph O'Rourke and Costin V\^\i lcu}

\date{\today}

\maketitle

\begin{abstract}
We prove that every positively-weighted tree $T$ can be realized
as the cut locus $C(x)$ of a point $x$ on a convex polyhedron $P$,
with $T$ weights matching $C(x)$ lengths.
If $T$ has $n$ leaves, $P$ has (in general) $n+1$ vertices.
We show there are in fact a continuum of polyhedra $P$ each realizing
$T$ for some $x \in P$.
Three main tools in the proof are properties of the star unfolding of $P$, 
Alexandrov's gluing theorem, and a cut-locus partition lemma.
The construction of $P$ from $T$ is surprisingly simple.
\end{abstract}


\section{Introduction}
\seclab{Introduction}

There is a long tradition of reversing, in some sense, the construction of
a graph $G$ from a geometric set.
The geometric set may be a point set, a polygon, or a polyhedron,
and the graph $G$ could be the Voronoi Diagram, the straight skeleton,
or the cut locus, respectively.
Reversing would start with, say,  the straight skeleton,
and reconstruct a polygon with that skeleton.
Here we start with the cut locus and construct polyhedra $P$
on which the cut locus is realized for a point $x \in P$.
(The cut locus is defined in Section~\secref{CL} below.)

The literature has primarily examined three models for the
graph $G$, often specialized (as here) to trees $T$:
\begin{enumerate}[label={(\arabic*)}]
\squeezelist
\item \emph{Unweighted tree}: The combinatorial structure of $T$, without further information. 
\item \emph{Length tree}: $T$ with positive edge weights representing Euclidean
lengths, and with given circular order of the edges incident to each node of $T$.
Called ``ribbon trees" in~\cite{cdlr-scrt-14}, and ``ordered trees" 
in~\cite{biedl2016realizing}.
\item \emph{Geometric tree}: Given by a drawing, i.e., coordinates of nodes,
determining lengths and angles.
\end{enumerate}

\noindent
Our main result is this:
\begin{thm}
\thmlab{main}
Given a length tree $T$ of $n$ leaves, we can construct a continuum of
star-unfoldings of convex polyhedra $P$ of $n+1$ vertices, each of which,
when folded, realizes $T$ as the cut locus $C(x)$ for a point $x \in P$.
Each star-unfolding can be constructed in $O(n)$ time.
\end{thm}
\noindent
Thus, every length tree is isometric to a cut locus on a convex polyhedron.

\subsection{Related Results}
The computer science literature is extensive, and we cite just
a few results:
\begin{itemize}
\squeezelist
\item Every unweighted tree can be realized as the Voronoi diagram of
a set of points in convex position~\cite{liotta2003voronoi}.
\item Every length tree can be realized as the furthest-point Voronoi
diagram of a set of points~\cite{biedl2016realizing}.
\item Every unweighted tree can be realized as the 
straight skeleton of a convex polygon,
and conditions for length-tree realization are 
known \cite{cdlr-scrt-14} \cite{aichholzer2015representing} \cite{biedl2016realizing}.
\end{itemize}
In all cases, the reconstruction algorithms are efficient:
either $O(n)$ or $O(n \log n)$ for trees of $n$ nodes.
Although all these results can be viewed as variations on realizing Voronoi diagrams,
and a cut locus is a subgraph of a Voronoi diagram,
it appears that prior work does not imply our results. 

Our inspiration derives from two results in the convexity literature:
\begin{itemize}
\squeezelist
\item Every length graph can be realized as a cut locus on a Riemannian 
surface~\cite{itoh2015every}. The result is non-constructive. 
\item Every unweighted tree can be realized as 
a cut locus on a doubly covered convex polygon,
and length trees can be realized on such polygons when several conditions are 
satisfied~\cite{itoh2004farthest}.
\end{itemize}


\section{Background}

In this section we describe the tools needed to prove our main theorem,
drawing heavily on our~\protect\cite{ov-teb-2020}.

\subsection{Cut Locus}
\seclab{CL}
The \emph{cut locus} $C(x)$ of a point $x$ on (the surface of) a convex polyhedron $P$ is the closure of the set of points to which there are more than one shortest path from $x$.
This concept goes back to Poincar\'e~\cite{p-lgsc-1905}, and has been studied algorithmically 
since~\cite{sharir1986shortest}
(under the name ``ridge tree").
Some basic properties and terminology:
\begin{itemize}
\squeezelist
\item $C(x)$ is a tree whose endpoints are vertices of $P$, and all vertices of $P$ are in $C(x)$.
\item Points interior to $C(x)$ of tree-degree $3$ or more 
we will call \emph{ramification points}.
\item The edges of $C(x)$ are \emph{geodesic segments} on $P$,
geodesic shortest paths between their endpoints ~\cite{aaos-supa-97}. 
\end{itemize}

\subsection{Star Unfolding}
The \emph{star unfolding} $S_P(x)$ of $P$ with respect to $x$ is formed by cutting
a shortest path from $x$ to every vertex of $P$
\cite{ao-nsu-92}~\cite{aaos-supa-97}.
This unfolds to a simple non-overlapping planar polygon 
$S=S_P(x)$ of $2n$ vertices: $n$ images $x_i$ of $x$, and $n$ images of the vertices $v_i$ of $P$.
The connection between the cut locus and the star unfolding is that
the image of $C(x)$ in $S$ is the restriction to $S$ of the Voronoi diagram of the images of $x$  \cite{ao-nsu-92}.
See Fig.~\figref{StarUnfCube}.
\begin{figure}
\centering
 \includegraphics[width=1.0\textwidth]{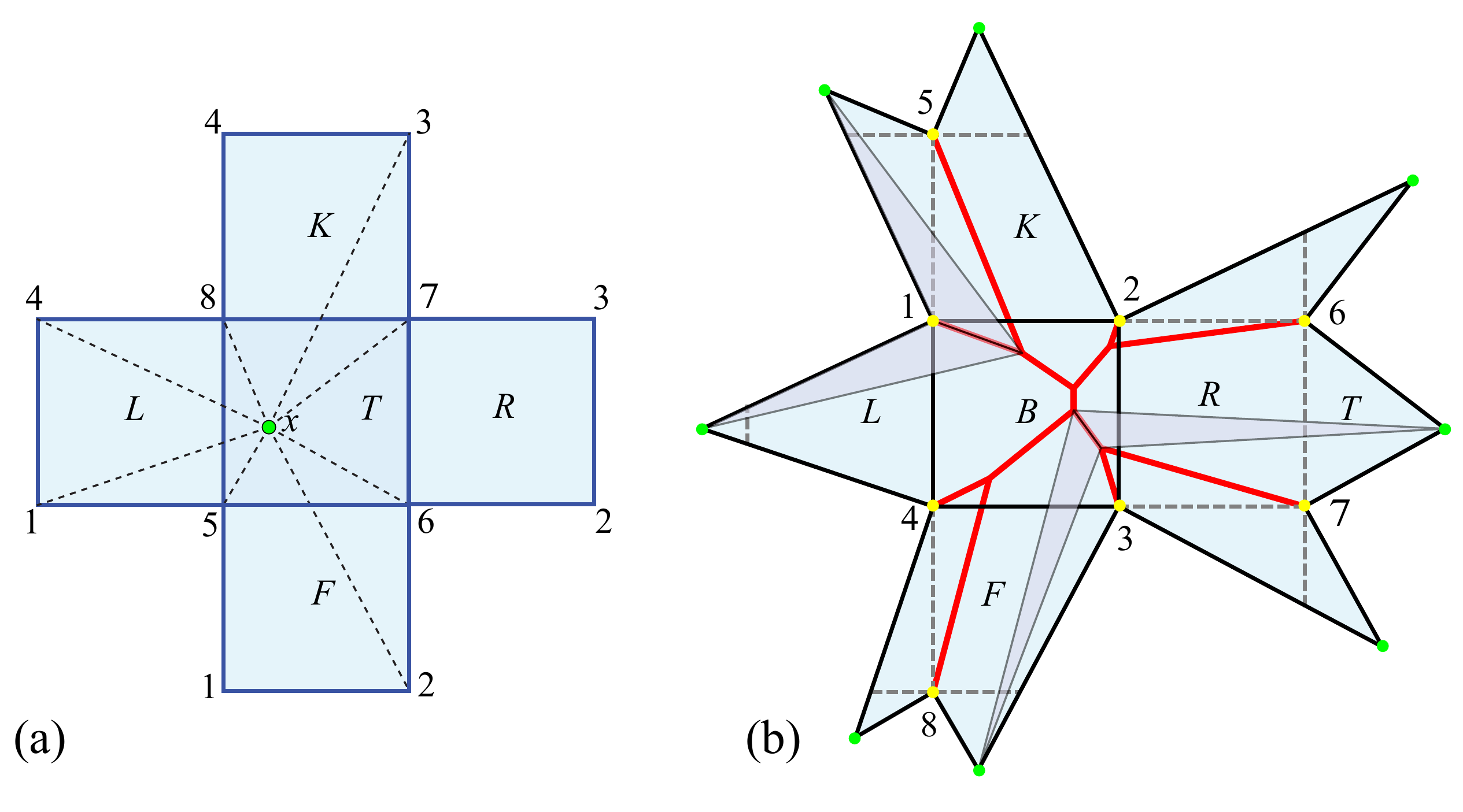}
\caption{(a)~Cut segments to the $8$ vertices of a cube from a point $x$
on the top face. T, F, R, K, L, B $=$ Top, Front, Right, Back, Left, Bottom.
(b)~The star-unfolding from $x$. The cut locus $C(x)$ (red) is the Voronoi
diagram of the $8$ images of $x$ (green).
Two pairs of fundamental triangles are shaded.}
\figlab{StarUnfCube}/
\end{figure}

\subsection{Alexandrov's Gluing Theorem}
We rely on Alexandrov's celebrated ``Gluing" Theorem~\cite[p.100]{a-cp-05}.

\begin{customthm}{AGT}
\label{AGT} 
If the boundaries of planar polygons are glued together (by identifying portions of the same length) such that
\begin{enumerate}[label={(\arabic*)}]
\item The perimeters of all polygons are matched (no gaps, no overlaps).
\item The resulting surface is a topological sphere.
\item At most $2\pi$ surface angle is glued at any point.
\end{enumerate}
Then the result is isometric to a convex polyhedron $P$,
possibly degenerated to a doubly-covered convex polygon. 
Moreover, $P$ is unique up to rigid motion and reflection. 
\end{customthm}
\noindent
The proof of this theorem is nonconstructive, and there remains no effective
procedure for constructing the polyhedron guaranteed to exist by this theorem.

\subsection{Fundamental Triangles}
The following lemma is one key to our proof.
See Fig.~\figref{StarUnfCube}(b).
\begin{lm}[Fundamental Triangles~\cite{inv-cfcp-12}]
\lemlab{FundamentalTriangles} 
For any point $x \in P$,
$P$ can be partitioned into flat triangles whose bases are edges of $C(x)$,
and whose lateral edges are geodesic segments from $x$ to 
the ramification points 
or leaves of $C(x)$. Moreover, those triangles are isometric to
plane triangles, congruent by pairs. 
\end{lm}
\noindent

The overall form of our proof of Theorem~\thmref{main} is to create
a star unfolding $S$ by pasting together $x_i$-apexed fundamental triangles straddling
each edge of $T$, and then applying Alexandrov's theorem 
to conclude that the folding of $S$ 
yields a convex polyhedron $P$ which realizes $C(x)$.

\subsection{Cut Locus Partition}

The last tool we need is a generalization of lemmas in~\cite{inv-cfcp-12}.
On a polyhedron $P$,
connect a point $x$ to a point $y \in C(x)$ by two geodesic segments
$\g, \g'$.
This partitions $P$ into two ``half-surface" digons $H_1$ and $H_2$.
If we now zip each digon separately closed by joining $\g$ and $\g'$,
AGT leads to two convex polyhedra $P_1$ and $P_2$.
The lemma says that the cut locus on $P$ is the ``join" of the cut loci on $P_i$.
See Fig.~\figref{Partition}.

\begin{lm}
\lemlab{Partition}
Under the above circumstances, the cut locus $C(x,P)$ of $x$ on $P$
is the \emph{join} of the cut loci on $P_i$:
$C(x,P) = C(x,P_1) \sqcup_y C(x,P_2)$, where $\sqcup_y$ joins the two cut loci at $y$.
And starting instead from $P_1$ and $P_2$, the natural converse holds as well.
\end{lm}

\begin{proof}
(\emph{Sketch}. See~\cite{ov-teb-2020} for a formal proof.)
All geodesic segments starting at $x$ into $H_i$
remain in $H_i$, because geodesic segments do not branch.
Therefore, $H_1$ has no influence on $C(x,P_2)$ and
$H_2$ has no influence on $C(x,P_1)$.
\end{proof}
\begin{figure}[htbp]
\centering
\includegraphics[width=0.75\linewidth]{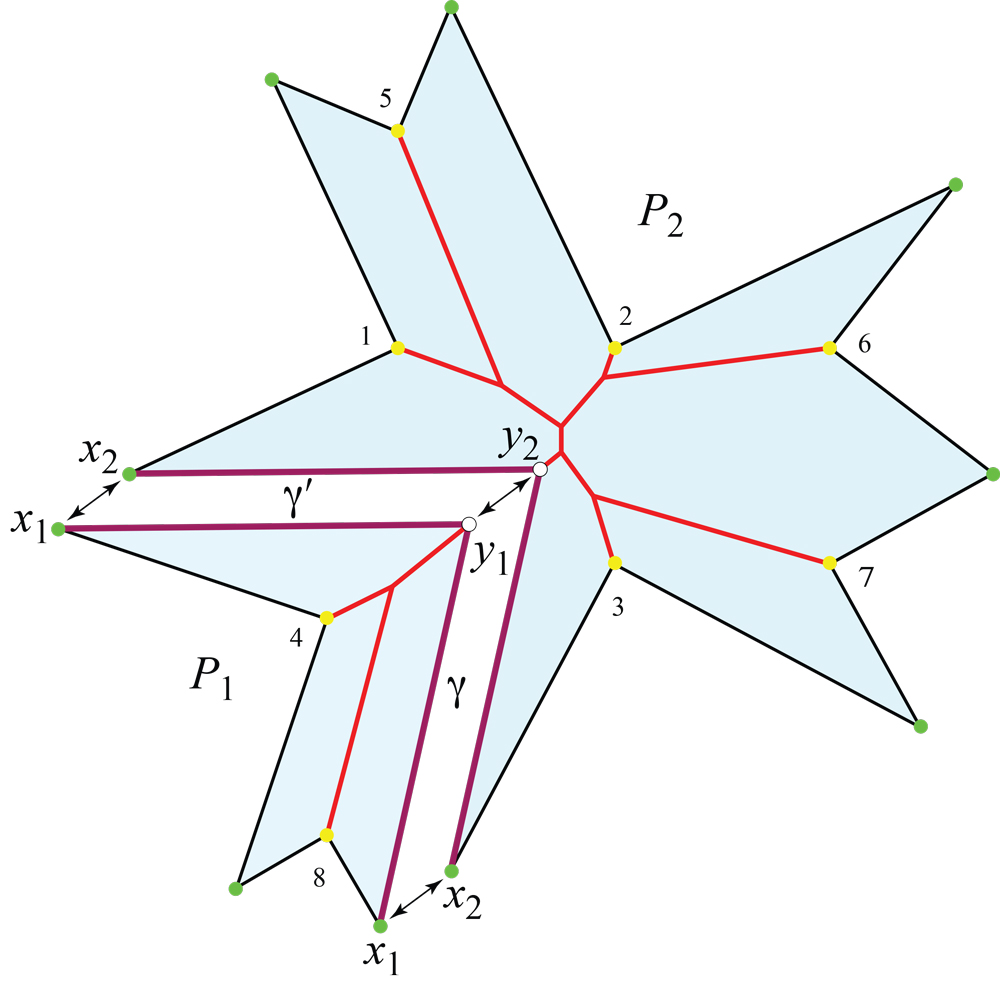}
\caption{Geodesic segments
$\g$ and $\g'$ (purple) connect $x{=}x_1{=}x_2$ to $y{=}y_1{=}y_2$.
$P_1$ folds to a tetrahedron, and $P_2$ to an $8$-vertex polyhedron, with $x$ and $y$ vertices in each.
$P_1$ and $P_2$ are cut open along geodesic segments from $x_i$ to $y_i$ and glued together to form $P$.
Based on the cube unfolding in Fig.~\protect\figref{StarUnfCube}(b).
}
\figlab{Partition}
\end{figure}

\section{Star-Tree}
If $T$ has a node of degree-$2$, then its incident edges may be merged
and their edge-lengths summed.
So henceforth we assume $T$ has no nodes of degree-$2$.
We start with $T$ a star-tree:
one central node $u$ of degree-$m$ with edges to nodes $u_1,u_2,\ldots,u_m$. 

A \emph{cone} in the plane is the unbounded region between two rays from its apex. 
Set $\lambda > L$ to be longer than $L$, the length of the longest edge of $T$.

We realize $T$ within a cone of apex angle $\a$, with $0 < \a \le 2 \pi$.%
\footnote{
Note that we allow $\a > \pi$; $\a=2\pi$ represents the whole plane.}
See Fig.~\figref{StarGraphLemma}.
Identify points $x_1$ and $x_{m+1}$ on the cone boundary, with
$|u x_i|=\lambda$.
Inside the cone, place $x$ images
$x_2,\ldots,x_m$, with each $|u x_i|=\lambda$.
Finally place $u_i$ so that $u u_i$ bisects 
$\angle(u x_i,u x_{i+1})$, $i=1,\ldots,m$.
Chose $u_i$ so that $|u u_i|$ matches $T$'s edge weights.
Finally, connect
$$
( u, x_1, u_2, x_2, \ldots, x_m, u_m, x_{m+1} )
$$
into a simple polygon.%
\footnote{If $\a=2\pi$, $x_1 = x_{m+1}$ and $u$ is interior to the polygon.}
\begin{figure}[htbp]
\centering
\includegraphics[width=0.70\linewidth]{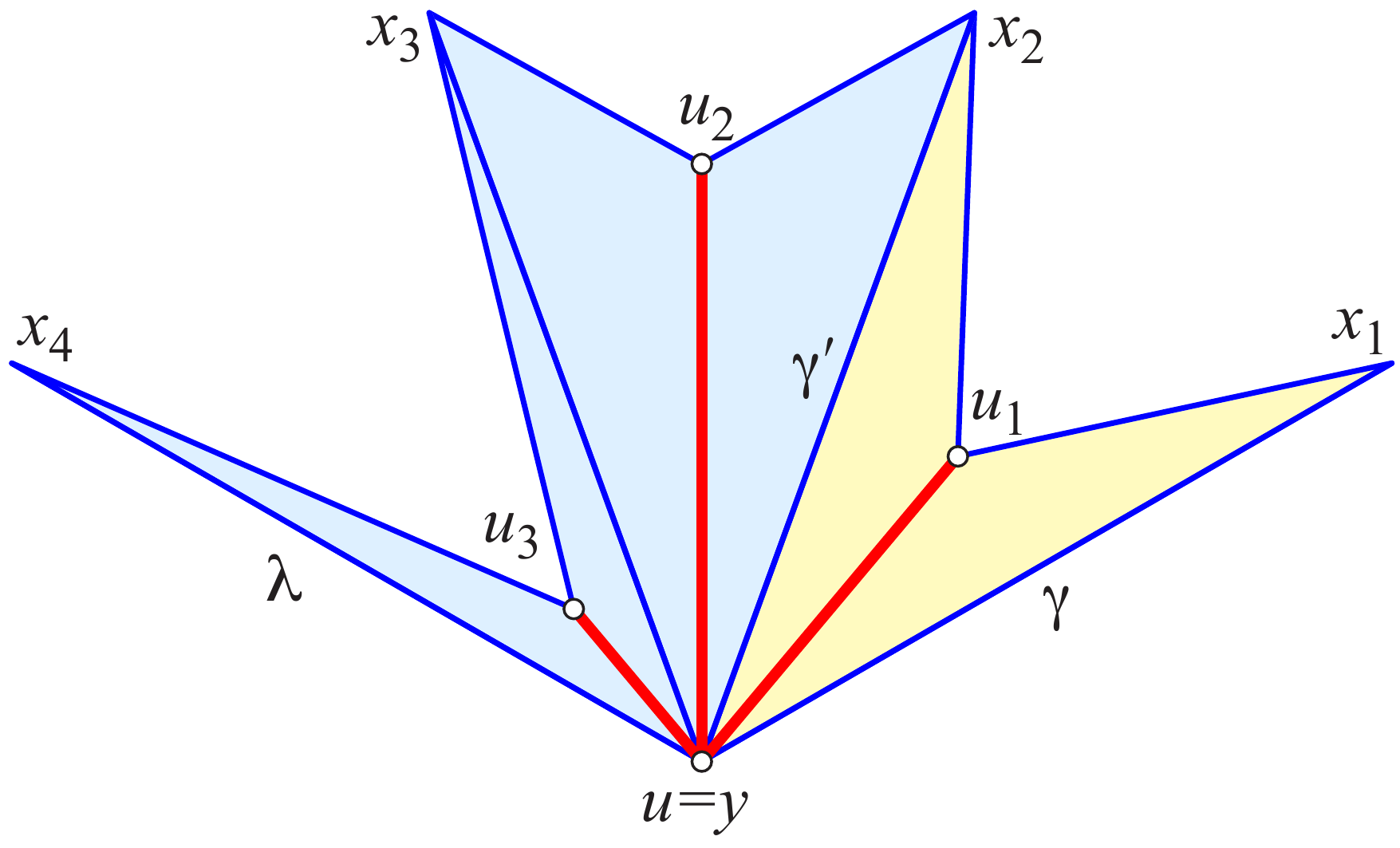}
\caption{$\a=120^\circ$, $m=3$, $T$ edge lengths $(2,3,1)$ (red), $\lambda=4$.
In the induction proof of Lemma~\protect\lemref{Star-T},
$\bar{P}(T_1)$ (yellow) is joined to $\bar{P}(T_2)$ (blue).}
\figlab{StarGraphLemma}
\end{figure}

Before we proceed with the proof, we emphasize that there are
several free choices in this construction, illustrated in Fig.~\figref{TwoParams}:
\begin{itemize}
\squeezelist
\item The angle $\a$ at the root is arbitrary.
\item The angular distribution of the $u x_i$ segments is arbitrary.
\item $\lambda > L$ needs to be ``sufficiently large" in a sense
we will quantify, but otherwise is arbitrary.
\end{itemize}
In general we will distribute 
$x_i$ equi-angularly.
Choosing $\a=2\pi$ results in a polyhedron of $n+1$ vertices;
for $\a< 2 \pi$, $u$ is an additional vertex.
\begin{figure}[htbp]
\centering
\includegraphics[width=1.0\linewidth]{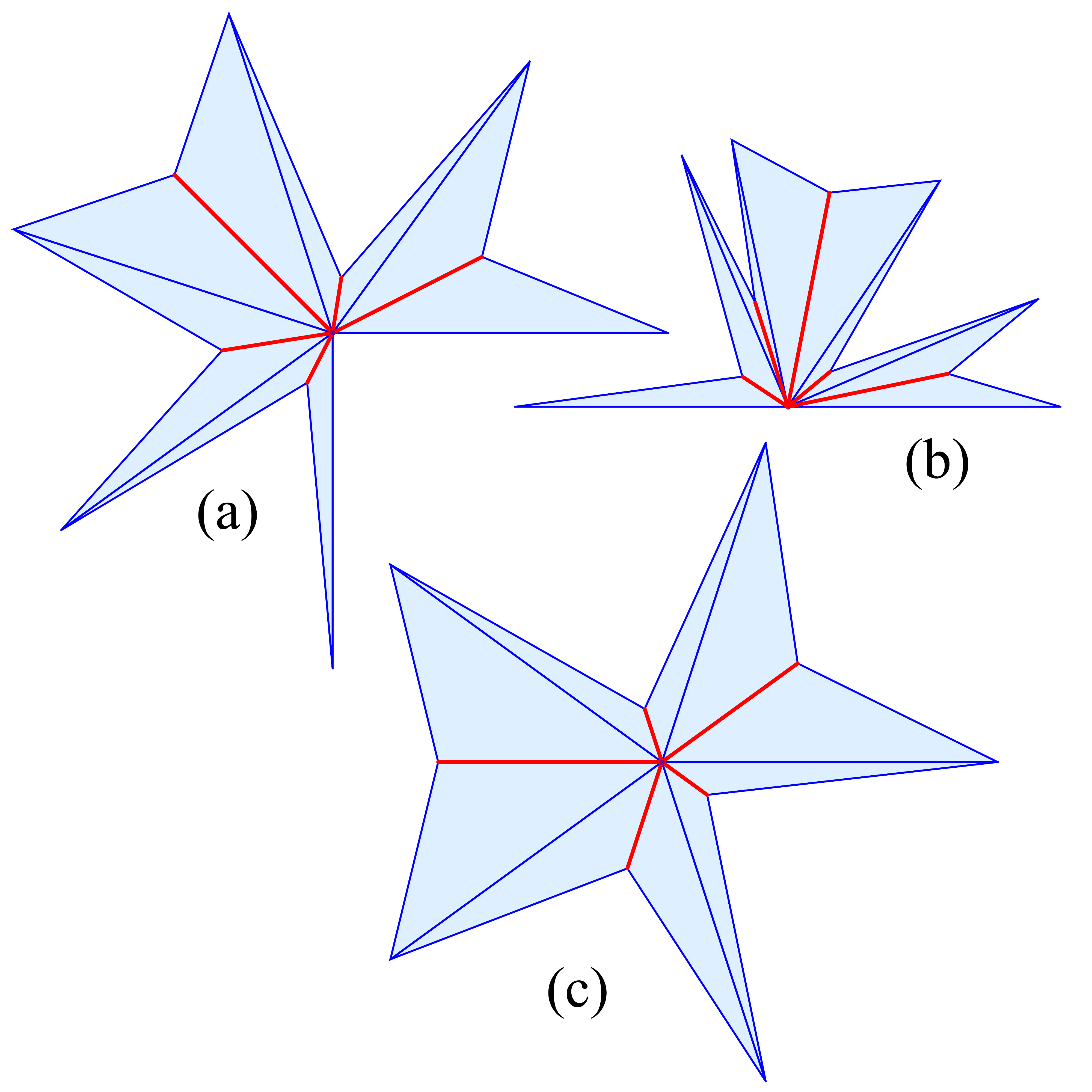}
\caption{Star-$T$ with edge lengths $(3, 1, 4, 2, 1)$.
(a)~$\a=270^\circ$, equiangular $x_i$, $\lambda=6$.
(b)~$\a=180^\circ$, random $x_i$, $\lambda=5$.
(c)~$\a=360^\circ$, equiangular $x_i$, $\lambda=6$.}
\figlab{TwoParams}
\end{figure}

We call the described star-$T$ construction a \emph{triangle packing}.
\begin{lm}
\lemlab{Star-T}
A triangle packing of star-$T$, for sufficiently large $\lambda$,
is the star-unfolding of a polyhedron $P$ with respect to a point $x$ such that $T=C(x)$.
\end{lm}
\begin{proof}
The proof is by induction on 
the number $m$ of edges of $T$, which is the degree of the root node $u$.
If $T$ is a single edge $e=u u_1$, then folding the twin triangles by creasing $e$ and joining $x_1$ and $x_2$, the two images of $x$,
leads to a doubly covered triangle with $x$ at the corner opposite $e$.
(See the yellow triangles in Fig.~\figref{StarGraphLemma}.)
Clearly $e=C(x)$ is the restriction of the Voronoi diagram of $x_1$ and $x_2$, the bisector of $x_1$ and $x_2$.
Concerning $\lambda$, we need that
the angle $\q_x$ incident to $x$ is at most $2\pi$.
In this base case, $\q_1 + \q_2 \le 2 \pi$ is immediate, 
so $\lambda$ is sufficiently large.

Now let $n>1$, and partition $T$ into $T_1 \cup_u T_2$ with root degrees $m_1$ and $m_2$ respectively, where $\cup_u$ indicates joining the trees at the root $u$.
We will use $\bar{P}(T)$ for the planar triangle packing for $T$, and $P(T)$ for the folded polyhedron.

We first address $\lambda$.
In order to apply AGT, we need that $\q_x$, the sum of the angles at the tips of the triangles---the images of $x$---is at most $2 \pi$.

Let $\lambda$ be the larger of $\lambda_1$ and $\lambda_2$ for $\bar{P}(T_1)$ and $\bar{P}(T_2)$ respectively, and stretch the smaller $\lambda_i$
so that they both share the same $\lambda$.
Form $\bar{P}(T)$ by joining the two now-compatible triangle packings.
Fixing $\a$ and the sector angles, it is clear that the angle at each triangle tip decreases monotonically as $\lambda$ increases.
So increase $\lambda$ as needed so that $\q_x \le 2 \pi$.
Somewhat abusing notation, call these possibly enlarged packings $\bar{P}(T_1)$, $\bar{P}(T_2)$ and $\bar{P}(T)$.

Now we aim to show that  $\bar{P}(T)$ is the star-unfolding of a polyhedron with $T=C(x)$.
Certainly  $\bar{P}(T)$ folds to a polyhedron $P$, because (a)~by construction the edges incident to $u_1,\ldots,u_k$ match in length,
and (b)~we have ensured that $\q_x \le 2 \pi$.
So Alexandrov's theorem applies.
Now identify $\g$ and $\g'$ on $P$ from $x$ to $y=u$ separating the surface into pieces corresponding to
$\bar{P}(T_1)$ and $\bar{P}(T_2)$, which fold to $P_1$ and $P_2$ respectively.
(Refer again to Fig.~\figref{StarGraphLemma}.)
By the induction hypothesis, $T_1 = C(x,P_1)$ and $T_2 = C(x,P_2)$.
Applying Lemma~\lemref{Partition}, we have
$C(x,P) = C(x,P_1) \sqcup_y C(x,P_2)$, where the two cut loci are joined at $y=u$.
And so indeed $T_1 \cup_u T_2 = T = C(x)$.
\end{proof}
\noindent
In Section~\secref{xFormula} we will calculate the needed $\lambda$ explicitly.

\section{General Length-Trees $T$}
\seclab{General}
We now generalize the above to arbitrary length-trees $T$, using the example in Fig.~\figref{Level23} as illustration.
\begin{figure}[htbp]
\centering
\includegraphics[width=1.0\linewidth]{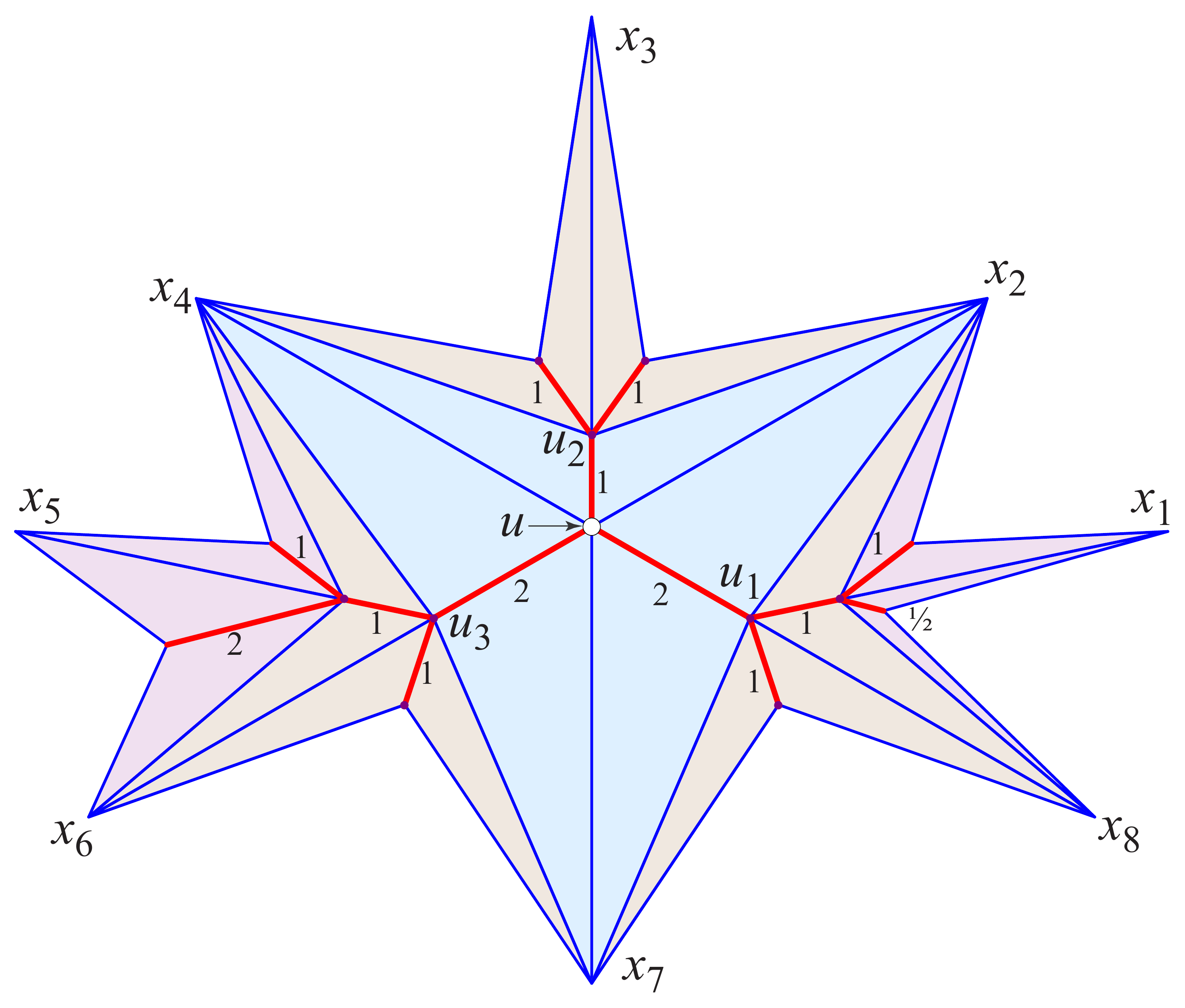}
\caption{Realization of a length-tree $T$ of height $3$, using
$\a=2\pi$ and $\lambda=L=5$. This polygon folds to a 
polyhedron of $9$ vertices: the $8$ leaves of $T$, and $x$.}
\figlab{Level23}
\end{figure}

Given $T$, select any node $u$ to serve as the root.
Fix any $\a$, and choose $\lambda$ to exceed the length $L$ of the longest path
from $u$ to a leaf in $T$.
Now create a triangle packing for $T$ as follows.

First create a triangle packing for $u$ and its immediate
children $u_1,\ldots,u_m$, just as previously described.
With $\lambda > L$,
the external angle at $u_i$ is strictly less than $\pi$, i.e., it forms a ``$V$-shape" there.
Call this a \emph{cup}, $c_i=(x_i, u_i, x_{i+1})$ with $\a_i$
external angle at $u_i$.
Let $u_i$ have children $u_{i1},u_{i2},\ldots,u_{il}$.
So $u_i, (u_{i1},\ldots,u_{il})$ is a star-tree.
Fill in the cup $c_i$ by inserting a triangle packing
for this sub-star-tree, with apex at $u_i$, angle $\a_i$,
and $\lambda$-length $| u_i x_i |$, the distance from $u_i$ to the tips of the cup.

After filling the $c_i$ cups for all the $u_i$ at level-$2$ of $T$, repeat
the process with level-$3$ of $T$, and so on.
Throughout the construction, the locations for $x_i$ remain fixed after 
their initial placement.
And with sufficiently large $\lambda$, all the cups form $V$-shapes.

Note that the triangles incident to an internal node $u_i$ of $T$
(neither the
root nor a leaf)
leave no gaps: they cover the $2\pi$ surrounding $u_i$.

\begin{lm}
\lemlab{Any-T}
A triangle packing for any $T$, as just described, for sufficiently large $\lambda$,
is the star-unfolding of a polyhedron $P$ with respect to a point $x$ such that $T=C(x)$.
\end{lm}
\begin{proof}
The proof is by induction, and parallels the proof of Lemma~\lemref{Star-T}
closely. Consequently, we only sketch the proof.

The base of the induction is a star-graph, settled by Lemma~\lemref{Star-T}.
Let $T$ be an arbitrary length-tree, and partition $T$ into two smaller trees $T_1$ and $T_2$
sharing the root $u$, so $T = T_1 \cup_u T_2$.
Select an $\a$ for $T$ and $\a_i$ for $T_i$, $i=1,2$, so that $\a= \a_1+\a_2$. 

By the induction hypothesis, $T_i$ can be realized in cups of angle $\a_i$.
Moreover, $\bar{P}(T_i)$ folds to $P_i$ and $T_i = C(x,P_i)$.
Stretch $\lambda_i$ as needed to allow $\bar{P}(T_1)$ to share $\lambda$ with $\bar{P}(T_2)$ at $u$, and stretch again so that $\q_x \le 2 \pi$.

Form $\bar{P}(T)$ by adjoining $\bar{P}(T_1)$ and $\bar{P}(T_2)$ at $u$, with cup apex $\a$.
Fold $\bar{P}(T)$ to $P$ by AGT. 
Apply Lemma~\lemref{Partition} to conclude $C(x,P) = C(x,P_1) \sqcup_y C(x,P_2)$, where $y=u$.
And so $T_1 \cup_u T_2 = T = C(x)$.
\end{proof}

\noindent
Note that all ramification points of $C(x)$ are flat on $P$, with $2 \pi$ incident surface angle.
If $\q_x$ is strictly less than $2 \pi$, then the source $x$ is a vertex on $P$.
If at the root, $\a < 2\pi$, then in addition $u$ is a vertex on $P$.
So $P$ has $n$, $n+1$, or $n+2$ vertices.

Lemmas~\lemref{Star-T} and~\lemref{Any-T}, together with Lemma~\lemref{xFormula} (below)
prove Theorem~\thmref{main}.
The construction of the triangle packing can be achieved in $O(n)$ time:
we are given the cyclic ordering of the edges incident to each
node, so sorting is not necessary, and the level-by-level packing construction is
proportional to the the number of edges.


\subsection{Total angle at $x$}
\seclab{xFormula}
We derive here a sufficient condition on the value of the parameter $\lambda$ for the root cup to guarantee that $\q_x \leq 2 \pi$.
\begin{lm}
\lemlab{xFormula}
If $T$ has $m$ edges, and $L$ is the longest path
from the root in $T$
then $\q_x \leq 2 \pi$
when
$$
   \lambda \ge L \left[ 1 +  \cot \left( \frac{\pi}{m} \right) \right] \;.
$$
\end{lm}

\begin{figure}[htbp]
\centering
\includegraphics[width=1.0\linewidth]{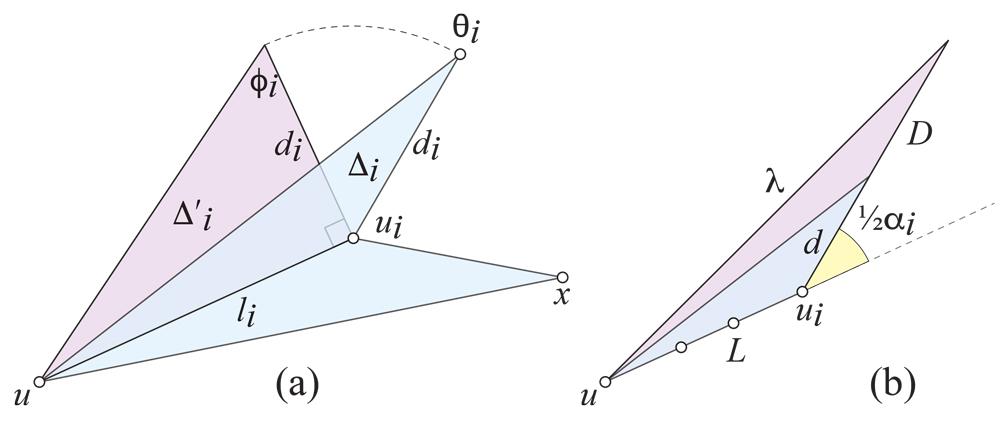}
\caption{(a)~$\q_i \le \f_i$. 
(b)~Increasing $\lambda$ increases $d$.
$L$ is the length of the longest path in $T$,
$D$ the target bound for $\lambda \ge L + D$.
}
\figlab{Formula}
\end{figure}

\begin{proof}
First we establish notation, illustrated in Fig.~\figref{Formula}(a).
As before, $u$ is the root of $T$, and let $u_i$ be a leaf of $T$, with edge $u u_i$ shared by twin triangles, one of which is $\triangle_i$. 
Let $\ell_i = |u u_i|$ and $d_i$ the distance from the tip of $\triangle_i$ to $u_i$, and $\q_i$ the angle at that tip, an image of $x$.
Because the fundamental triangles come in pairs, and there are $m$ edges, we have that the total angle at $x$ satisfies $\q_x = 2\sum \q_i$.

Consider now a right triangle $\triangle'_i$ having the same base as $\triangle_i$
and height $d_i$, and denote by $\f_i$ its angle opposite to the base $uu_i$.
Then $\f_i = \arctan \left(\dfrac{\ell_i}{d_i} \right)$
and $\q_i \leq \f_i$; 
see again Fig.~\figref{Formula}.

Because $\arctan$ is an increasing function, we can obtain an upper bound by replacing $\ell_i$ with the longest edge length $\ell$,
and replacing $d_i$ by the shortest of the $x u_j$ diagonals, call it $d$: ${\ell}= \max_i {\ell_i}$, $d= \min_i d_i$.
So $\f_i \le \arctan \left( \dfrac{\ell}{d} \right)$, and therefore
$$
\q_x \;=\; 2\sum \q_i \;\le\; 2\sum \f_i \;\le\; 2 m \arctan \left( \frac{\ell}{d} \right) \; .
$$
The expression $2 m \arctan \left( \dfrac{\ell}{d} \right)$ decreases as $d$ increases.
For it to evaluate to at most $2 \pi$, we must have 
$$
d \;\ge\; \ell \cot \left( \frac{\pi}{m} \right) \;.
$$
The longest path $L$ from root to leaf is at least as long as the longest edge, $L \ge \ell$, so this bound will more than suffice:
$$
d \;\ge\; L \cot \left( \frac{\pi}{m} \right) = D \;.
$$
Now we show that if $\lambda$ is long enough,
then $d \ge D$.
Consider Fig.~\figref{Formula}(b), where $L$ is the longest path
from $u$ to a leaf $u_i$. Let the external angle at $u_i$ be $\a_i$.
By definition, there is some index $j$ such that $d=d_j$.
The triangle inequality directly implies $\lambda  \le \ell_j  + d_j \le  L + d$.

With $L$ and $\a_i$ fixed, increasing $\lambda$ increases $d$.
If we substitute the needed lower bound $D$ for $d$ in the expression
(see Fig.~\figref{Formula}(b)), then
$\lambda \ge L + D$
forces $d \ge D$.
Explicitly
$$
\lambda \ge L \left[ 1 +  \cot \left( \frac{\pi}{m} \right) \right] \;
$$
suffices to guarantee that $\q_x \le 2 \pi$.
\end{proof}

\noindent
The bound---approximately $L(1+m / \pi)$---is far from tight.
For example, in Fig.~\figref{Level23}, $m=13$ and $L=5$ leads to $\lambda \ge 26$, 
but $\lambda=5$ (illustrated) leads to $\q_x \approx 167^\circ$,
and so easily suffices.

\section{Remarks}

One further example is shown in Fig.~\figref{Level5}.
It is the star unfolding of a polyhedron of $49$ vertices,
whose resemblance to a fractal suggests there might at 
a deeper connection.
We will only mention that fractals play a role in the folding of specific convex polyhedra in~\cite{uehara2020common},
and fractal cut loci on on $k$-differentiable Riemannian and Finslerian spheres are 
shown in~\cite{itoh2016riemannian} to exist for any $2 \leq k < \infty$.

\begin{figure}[htbp]
\centering
\includegraphics[width=1.0\linewidth]{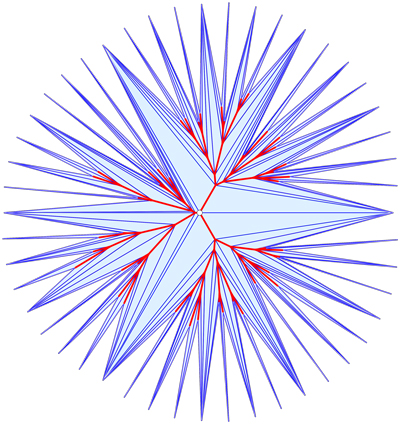}
\caption{Regular degree-$3$ tree, random edge lengths, $\a=2\pi$, $n=48$,
$\q_x \approx 317^\circ$.}
\figlab{Level5}
\end{figure}

A natural question is whether 
geometric trees---drawings embedded in the plane, and so providing
angles between adjacent edges---can be
realized as cut loci on convex polyhedra.
Certainly not all geometric trees are realizable, for there are constraints on the angles: 
around a ramification point,
no angle can exceed $\pi$, and the angles must sum to $\le 2 \pi$.
And the sum of the curvatures at the $u_i$ and at $x$
must be $4\pi$ to satisfy the Gauss-Bonnet theorem.

We leave it as an open problem to characterize those  geometric trees that are realizable as the cut locus on a convex polyhedron.

\bibliographystyle{alpha}
\bibliography{CutLocReal}
\end{document}